\documentclass[runningheads]{llncs}
\usepackage[T1]{fontenc}
\usepackage[utf8]{inputenc}
\usepackage{amsmath,amssymb,amsfonts,mathtools}
\usepackage{booktabs}
\usepackage{xspace}
\usepackage{graphicx}
\usepackage[hidelinks]{hyperref}
\usepackage[nameinlink,noabbrev]{cleveref}
\usepackage{orcidlink}

\newcommand\ma[1]{\ensuremath{\mathcal{#1}}}
\newcommand\mb[1]{\ensuremath{\mathbb{#1}}}

\newcommand\mf[1]{\ensuremath{\mathfrak{#1}}}

\newcommand{\PiP}{\Pi}
\newcommand{\D}{\ma D}
\newcommand{\Dx}{\ma D^{\mathsf{x}}}
\newcommand{\Dn}{\ma D^{\mathsf{n}}}
\newcommand{\EDB}{\ensuremath{\Sigma_{\ma E}}}
\newcommand{\IDB}{\ensuremath{\Sigma_{\ma I}}}
\newcommand{\adom}{\mathsf{adom}}

\newcommand{\HBIDB}{\ensuremath{\mf{B}^{\IDB}}}
\newcommand{\trnum}{\operatorname{tr}}
\newcommand{\Supp}{\mb S}
\newcommand{\Hyp}{\mathcal{H}}
\newcommand{\TC}{\Pi_{\mathsf{tc}}}
\newcommand{\Path}{\mathsf{Path}}
\newcommand{\Edge}{\mathsf{Edge}}
\newcommand{\Goal}{\mathsf{Goal}}
\newcommand{\Inf}{\infty}

\newcommand{\UCQ}{\mathsf{UCQ}}
\newcommand{\Destroy}{\mathsf{Destroy}}
\begin{document}
\title{Causality and Minimal Supports in Recursive
Datalog \thanks{To appear in the Proceedings of the 10th International Joint Conference on Rules and Reasoning (RuleML+RR 2026), LNCS, Vilnius, Lithuania.}}
\titlerunning{Causality and Supports in Datalog}
%
\author{Ratan Bahadur Thapa\inst{1}\thanks{Corresponding author: ratan.thapa@ki.uni-stuttgart.de} \orcidlink{0009-0000-2368-5928} \and
Steffen Staab\inst{1,2}\orcidlink{0000-0002-0780-4154} }

\authorrunning{R. B. Thapa and S. Staab}

%
\institute{Institute for Artificial Intelligence, University of Stuttgart,  Germany
\and
 University of Southampton, United Kingdom}
\maketitle  
\begin{abstract}
Explaining an inferred fact under rule evaluation can require identifying the inclusion-minimal input sets that suffice for the inference and the deletions that make the fact disappear. For a fixed union of conjunctive queries, every minimal support is bounded by the query body. For recursive rules, the same answer may depend on large supports, and the number of minimal supports may be exponential in the input. We study the gap through deletion-based explanation, using inclusion-minimal endogenous
input facts that entail the atom together with fixed background facts. We organize these supports as a hypergraph and prove that it determines actual causes, counterfactual causes, responsibility, and deletion robustness. The resulting view separates  nonrecursive queries from recursive
Datalog at the level of minimal input explanations. For positive-length
reachability, minimal supports are exactly simple directed paths, and deletion robustness is the minimum directed edge cut. We also prove invariance under fixed-goal equivalent positive Datalog programs and an NP-hardness calibration for the robustness threshold problem.
\end{abstract}

\section{Introduction}
When a database query returns a surprising answer, the natural explanatory
question is not only how the answer was derived, but also which input facts keep
the answer true and which deletions would make it disappear. Lineage and
provenance address the first part by
describing how an answer can be derived \cite{buneman2001and,green2007provenance}. 
Query-answer causality addresses the deletion part of the explanation
\cite{meliou2010causality,meliou2010complexity}. It distinguishes between
exogenous facts, which form the fixed background context, and endogenous facts,
which are allowed to be deleted.
An endogenous fact is an actual cause if
deleting it, possibly after deleting other endogenous facts while preserving the
answer, makes the answer false. We study the same
deletion-based questions for positive Datalog entailments, where answers are
derived by a least-fixpoint computation rather than by a single nonrecursive
query \cite{abiteboul1995foundations}.

The question is important because recursive rule languages are central in deductive databases, rule-based knowledge representation, and ontology-style reasoning, and answers derived using them often need to be explained in terms
of the input facts on which they depend. Consider a directed graph encoded by facts $\Edge(a,b)$. Throughout the paper,
let $\TC$ be the positive-length transitive-closure program with rules
\[\Path(x,y) \leftarrow \Edge(x,y),\quad
\Path(x,y) \leftarrow \Edge(x,z), \Path(z,y)\]
 deriving $\Path(x,y)$ from $\Edge(x,y)$, and deriving $\Path(x,y)$ from $\Edge(x,z)$ and $\Path(z,y)$. If several independent chains of rule application support the entailment of $\Path(s,t)$, then the answer survives exactly when at least one chain remains after deletion. Thus, the relevant explanation is not one derivation tree, but the family of minimal endogenous edge sets that still entail $\Path(s,t)$, in the sense of database witnesses, why-provenance, and lineage \cite{buneman2001and,green2007provenance}. We organize this family as a hypergraph whose vertices are deletable facts and whose hyperedges are minimal sufficient sets; we call it the support hypergraph. 

We follow database-causality terminology \cite{meliou2010causality,meliou2010complexity}: a \emph{counterfactual cause} is a
fact whose deletion alone falsifies the answer; the counterfactual refers to
the deletion test, not to the fact itself. An endogenous support is a set of endogenous facts that entails the answer
together with the fixed background facts; it is inclusion-minimal, or simply
minimal, if no proper subset has the same property. A fact is an actual cause if it belongs to some such support, and it is a counterfactual cause if it belongs to every support. The robustness radius is the minimum number of deletions needed to intersect all supports; responsibility is measured by the size of a smallest contingency deletion that makes a fact counterfactual cause.
\begin{example}\label{example_intro}
Consider a graph database in which all edge facts are endogenous, i.e., all of them may be deleted:
\[ e_1=\Edge(s,a),  e_2=\Edge(a,t), e_3=\Edge(s,b),
  e_4=\Edge(b,t), e_5=\Edge(a,b).\]
For $\TC$ and the answer $\Path(s,t)$, the minimal endogenous supports are
\[S_1=\{e_1,e_2\},\quad
  S_2=\{e_3,e_4\}\quad\text{and}\quad
  S_3=\{e_1,e_5,e_4\}.\]
Thus, the support hypergraph has vertices $e_1,\ldots,e_5$ and hyperedges $S_1,S_2,S_3$. Every edge fact is an actual cause, because each occurs in some support. No edge fact is a counterfactual cause, because no edge occurs in all three supports. The deletion robustness radius is $2$: one deletion cannot hit all three supports, while $\{e_1,e_4\}$ hits $S_1,S_2,S_3$ and destroys all $s$-$t$ paths, i.e., $\Path(s,t)$. The same hypergraph also determines responsibility: under the standard causality convention, a fact with smallest contingency size $k$ has responsibility
$1/(k+1)$. For instance, $e_1$ has a one-fact contingency, such as $\{e_3\}$, so its responsibility is $1/2$. By contrast, $e_5$ occurs only in $S_3$; preserving $S_3$ while hitting the supports avoiding $e_5$ requires deleting both $e_2$ and $e_3$, so its responsibility is $1/3$.
The example shows why we analyze
minimal input supports rather than derivation multiplicity: deletion behavior
depends on which fact sets still entail the answer, not on how many proof trees
derive it.
\end{example}

The main difficulty is that recursion breaks the bounded-witness behavior familiar from conjunctive queries (CQs). For a fixed CQ, each valuation uses only the atoms in the query body; therefore, every minimal endogenous support has bounded cardinality. A fixed union of CQs (UCQs) consequently has only polynomially many minimal supports in the database size. Recursive Datalog behaves differently. Already for the fixed program $\TC$, a minimal support for $\Path(s,t)$ can be a simple path whose length grows with the database, and a linear-size graph can contain exponentially many simple $s$-$t$ paths. Enumerating proof trees is not an adequate substitute: recursive programs may generate many derivations, including cyclic proof-tree unfoldings, whereas deletion robustness depends only on which endogenous facts intersect all minimal supports.

Prior work provides the causal and provenance foundations but does not isolate this support-level structure for recursive Datalog. Query-answer causality defines actual causes, counterfactual causes, contingencies, and responsibility for database answers \cite{meliou2010causality,meliou2010complexity}, building on structural-model causality and responsibility \cite{halpern2005causes,chockler2004responsibility}. 
Provenance and lineage describe witnesses, annotations, circuits, and derivations \cite{buneman2001and,green2007provenance,deutch2014circuits}. Resilience, deletion propagation, and how-to queries study related tuple-deletion or update problems, particularly for CQs and related nonrecursive classes \cite{salimi2015query,salimi2016causes,bertossi2017causes,freire2015complexity}.
Our focus is different: for positive Datalog entailment, we study its
inclusion-minimal deletable input supports. The resulting support hypergraph
captures causality, responsibility, and deletion robustness, separates fixed
UCQs from recursive Datalog by support structure, and distinguishes input
supports from derivation multiplicity.

Our main contribution is the isolation of the minimal-support structure
underlying deletion-based explanation for positive Datalog entailments. We: (i) formalize the support hypergraph of a derived atom over a database whose facts are divided into fixed background facts and deletable facts; (ii) prove that this support hypergraph determines actual causes, counterfactual causes, responsibility, and deletion robustness; (iii) then show that recursion fundamentally changes the support structure: fixed UCQs have bounded support cardinality and polynomially many minimal supports, while positive-length transitive closure can generate unbounded-size and exponentially many minimal supports;
and (iv) finally show that support hypergraphs depend on the goal semantics, not on
the chosen positive Datalog formulation: fixed-goal equivalent programs induce
the same minimal supports. For all-endogenous reachability with distinct
endpoints, deletion robustness becomes exactly the minimum directed
$s$-$t$ edge cut. We provide complete proofs of all stated results in the extended version~\cite{thapa2026causalityminimalsupportsrecursive}.

\section{Positive Datalog}
A finite relational signature $\Sigma$ is partitioned into extensional predicates $\EDB$ and intensional predicates $\IDB$, each with fixed arity. We allow nullary intensional predicates for Boolean Datalog queries. Constants are drawn from a countably infinite set, and terms are variables or constants. A database instance $\D$ is a finite set of ground atoms over $\EDB$. A positive Datalog rule is an expression
\[ H(\bar t) \leftarrow B_1(\bar t_1),\ldots,B_m(\bar t_m),\]
where $m\geq 0$, $H$ is an intensional predicate symbol of arity $k$, each
$B_i$ is a predicate symbol of arity $k_i$ over $\EDB \cup \IDB$, $\bar t$ is
a $k$-tuple of terms, each $\bar t_i$ is a $k_i$-tuple of terms, and every
variable occurring in the rule occurs in the body.

A positive Datalog program $\PiP$ is a finite set of such function-free rules. 
We use active-domain grounding. Let $\adom(\Pi,\mathcal D)$ be the set of
constants occurring in $\Pi$ or $\mathcal D$. The intensional Herbrand base
$\HBIDB_{\PiP,\D}$ is the finite set of all ground atoms
$P(\bar a)$ such that $P \in \Sigma_I$ and every constant in $\bar a$ belongs to
$\adom(\Pi,\mathcal D)$, i.e.,
\[
\HBIDB_{\PiP,\D}
=\{\,P(c_1,\ldots,c_k)
\mid P \in \IDB,\ \operatorname{arity}(P)=k,\ 
c_1,\ldots,c_k \in \adom(\PiP,\D)\,\}.
\]

The immediate-consequence operator $T_{\PiP,\D}$ maps a set
$I\subseteq \HBIDB_{\PiP,\D}$ of derived intensional atoms to the intensional
atoms obtained by one rule application over $\D\cup I$. Since $\PiP$ is
positive, $T_{\PiP,\D}$ is monotone. Starting with $I_0=\emptyset$, the Kleene sequence
$I_{i+1}=I_i\cup T_{\PiP,\D}(I_i)$ is increasing and reaches a least fixpoint
after finitely many steps, because $\HBIDB_{\PiP,\D}$ is finite. We write
$\PiP\cup\D\models A$ if $A$ is an extensional fact in $\D$ or an intensional
fact in that least fixpoint ~\cite{abiteboul1995foundations}. Entailment is always evaluated over the active domain of the program and the
retained extensional instance; atoms using constants outside that domain are not
entailed.

\subsection{Endogenous Tuples}

Database causality partitions input tuples into endogenous tuples, which may be candidate causes, and exogenous tuples, which form the background context \cite{meliou2010causality,meliou2010complexity}. 

\begin{definition}
\label{def:endogenous-instance}
Let $\PiP$ be a positive Datalog program over $\EDB\cup\IDB$. A partitioned instance for $\PiP$ is a triple $(\Dx,\Dn,A)$, where $\Dx$ and $\Dn$ are disjoint finite sets of ground atoms over $\EDB$, and $A$ is a ground atom over $\EDB\cup\IDB$ whose constants belong to $\adom(\PiP,\Dx\cup\Dn)$. We call $\Dx$ the exogenous database, $\Dn$ the endogenous database, and $\D=\Dx\cup\Dn$ the input database. Deletions range only over $\Dn$.
\end{definition}

For $\Delta\subseteq \Dn$, we write $\D\setminus\Delta$ as shorthand for $\Dx\cup(\Dn\setminus\Delta)$. Thus, a deletion removes only endogenous facts and leaves the exogenous database unchanged.

\subsection{Endogenous Supports}
\label{sec:supports}
We now isolate the inclusion-minimal endogenous fact sets sufficient for entailment. The resulting support hypergraph is the abstraction used for causality and robustness.

\begin{definition}
\label{def:support}
Let $\PiP$ be a positive Datalog program, let $(\Dx,\Dn,A)$ be a partitioned instance for $\PiP$, let $\D=\Dx\cup\Dn$, and assume $\PiP\cup\D\models A$. A set $S\subseteq\Dn$ is an endogenous support for $A$ if
\[
  \PiP\cup\Dx\cup S\models A.
\]
It is a minimal endogenous support if no proper subset $S'\subsetneq S$ satisfies $\PiP\cup\Dx\cup S'\models A$. 
When the partition $\D=\Dx\cup\Dn$ is clear from the partitioned instance, the family of all \emph{minimal endogenous supports} is
\[\Supp_{\PiP,\D,A} =
  \{\,S\subseteq\Dn \mid
      \PiP\cup\Dx\cup S\models A
      \text{ and }
      \forall S'\subsetneq S,\ 
      \PiP\cup\Dx\cup S'\not\models A
    \,\}.\]
The hypergraph of minimal supports for $A$ is the finite hypergraph
\[
  \Hyp_{\PiP,\D,A}
  =
  (\Dn,\Supp_{\PiP,\D,A}),\]
where vertices are the endogenous facts $\Dn$ and  hyperedges are the minimal endogenous supports for $A$.
\end{definition}

If $\PiP\cup\Dx\models A$, then $\emptyset$ is an endogenous support and the unique minimal endogenous support,  i.e., $\Supp_{\PiP,\D,A}=\{\emptyset\}$. In this exogenous-entailment case, no deletion of endogenous facts can destroy $A$, no endogenous fact is an actual or counterfactual cause, and the deletion robustness radius is $\infty$.

\begin{proposition}
\label{prop:decomposition}
Let $\PiP$ be a positive Datalog program, let $(\Dx,\Dn,A)$ be a partitioned instance for $\PiP$, let $\D=\Dx\cup\Dn$, and assume $\PiP\cup\D\models A$. Then, $\Supp_{\PiP,\D,A}$ is a finite antichain under set inclusion, and for every $E\subseteq\Dn$,
\[
  \PiP\cup\Dx\cup E\models A
  \quad\Longleftrightarrow\quad
  \exists S\in\Supp_{\PiP,\D,A}\text{ such that }S\subseteq E.
\]
\end{proposition}

\begin{proof}
The family $\Supp_{\PiP,\D,A}$ is a subset of $2^{\Dn}$. Since $\Dn$ is finite, $\Supp_{\PiP,\D,A}$ is finite. If $S_1,S_2\in\Supp_{\PiP,\D,A}$ and $S_1\subsetneq S_2$, then $S_2$ is not inclusion-minimal among supports, because $S_1$ is a proper supporting subset. Thus, $\Supp_{\PiP,\D,A}$ is an antichain.

Let $E\subseteq\Dn$. Assume first that $\PiP\cup\Dx\cup E\models A$. The finite set
$\mathcal{F}_E=\{\,F\subseteq E : \PiP\cup\Dx\cup F\models A\,\}$
contains $E$. Choose an inclusion-minimal member $S$ of $\mathcal{F}_E$. If some $S'\subsetneq S$ satisfied $\PiP\cup\Dx\cup S'\models A$, then $S'\in\mathcal{F}_E$, contradicting the choice of $S$. Thus, $S\in\Supp_{\PiP,\D,A}$ and $S\subseteq E$.

Conversely, assume that $S\in\Supp_{\PiP,\D,A}$ and $S\subseteq E$. If $A$ is extensional, then $A\in\Dx\cup S$. Since $S\subseteq E$, we have $A\in\Dx\cup E$, and thus $\PiP\cup\Dx\cup E\models A$. If $A$ is intensional, let $(I_i)_{i\geq 0}$ and $(J_i)_{i\geq 0}$ be the Kleene sequences for $\PiP$ over $\Dx\cup S$ and $\Dx\cup E$, respectively. Since $S\subseteq E$, induction on $i$ gives $I_i\subseteq J_i$: the base case is empty, and every ground rule instance whose body is true in $(\Dx\cup S)\cup I_i$ is also true in $(\Dx\cup E)\cup J_i$. Therefore, the least fixpoint over $\Dx\cup S$ is contained in the least fixpoint over $\Dx\cup E$, so $\PiP\cup\Dx\cup E\models A$.
\end{proof}

Proposition~\ref{prop:decomposition} is the semantic basis for the paper: grounded atom $A$ survives over a retained set $E\subseteq\Dn$ exactly when $E$ contains a minimal support.

\section{Recursive Datalog}\label{sec:recursive-supports}
We compare fixed UCQs with one fixed recursive Datalog program. A CQ has the form
$q(\bar{x}) = \exists\bar{y}\,\bigwedge_{i=1}^{m} R_i(\bar{z}_i)$, where each $\bar{z}_i$ uses variables from $\bar{x}\cup\bar{y}$ \cite{abiteboul1995foundations}. For a database $\D$ and tuple $\bar a$,\, $\D\models q(\bar a)$ holds if there is a valuation $h$ with $h(\bar{x})=\bar a$ and $R_i(h(\bar z_i))\in \D$ for every $i$. A UCQ is a finite disjunction of CQs with the same free
variables.

For a $\UCQ$ $q$ over the extensional signature $\EDB$, a disjoint partition $\D=\Dx\cup\Dn$ of a finite $\EDB$-database, and an answer $\bar a$ with $\D\models q(\bar a)$, a set $S\subseteq\Dn$ is an endogenous support for $q(\bar a)$ if
$\Dx\cup S\models q(\bar a)$.
It is minimal if no proper subset has this property. For a fixed $\UCQ$ $q$, let $m(q)$ be the maximum number of atoms in a disjunct of $q$, and let $\Supp_{q,\D,\bar a}$ denote the minimal endogenous supports
for $q(\bar a)$.

\begin{theorem}\label{thm:separation}
Following holds.
\begin{enumerate}
    \item For every fixed $\UCQ$ $q$, every disjoint partition $\D=\Dx\cup\Dn$, every answer $\bar a$, and every $S\in\Supp_{q,\D,\bar a}$,
\[|S|\leq m(q)\quad\text{and}\quad
  |\Supp_{q,\D,\bar a}|\leq \sum_{i=0}^{m(q)}\binom{|\Dn|}{i}.\]
  \item For the fixed positive-length transitive-closure program $\TC$, there
is a family of instances $(\emptyset,\D_r,A_r)_{r\geq 1}$, with exogenous
database empty and endogenous database $\D_r$, such that
\[
\max_{S\in\Supp_{\TC,\D_r,A_r}} |S|=\Theta(|\D_r|)
\quad\text{and}\quad
|\Supp_{\TC,\D_r,A_r}|=2^{\Omega(|\D_r|)}.
\]
\end{enumerate}
\end{theorem}

 In Theorem~\ref{thm:separation}, asymptotic notation is with respect to database size: $\Theta(|\D_r|)$ means linear in $|\D_r|$, and
$2^{\Omega(|\D_r|)}$ means exponential in $|\D_r|$. The first part of the theorem gives polynomially many minimal supports in $|\Dn|$ for
fixed $\UCQ$s. The second part shows that a fixed recursive program can have minimal supports of unbounded cardinality and exponentially many such
supports.

\begin{proof}
Let $S\in\Supp_{q,\D,\bar a}$. Since $\Dx\cup S\models q(\bar a)$, there are a disjunct $q_j$ of $q$ and a valuation $h$ such that $h$ maps the free variables of $q_j$ to $\bar a$ and maps every body atom of $q_j$ to a fact of $\Dx\cup S$. Let
\[
  W=\{\,R(h(\bar z))\in S \mid R(\bar z)\text{ is a body atom of }q_j\,\}.
\]
Every body atom of $q_j$ is mapped either to an exogenous fact in $\Dx$ or to a fact in $W$. Thus, the same valuation $h$ witnesses $\Dx\cup W\models q(\bar a)$. Since $W\subseteq S$ and $S$ is inclusion-minimal among endogenous supports, $S=W$. The chosen disjunct has at most $m(q)$ body atoms, so $|S|=|W|\leq m(q)$. Therefore, every minimal support is a subset of $\Dn$ of size at most $m(q)$, and the number of such subsets is bounded by $\sum_{i=0}^{m(q)}\binom{|\Dn|}{i}$.

For the recursive case, consider the fixed positive-length transitive-closure program
\[
\begin{array}{rcl}
\Path(x,y) \leftarrow \Edge(x,y),\quad
\Path(x,y) \leftarrow \Edge(x,z),\Path(z,y).
\end{array}
\tag{$\TC$}
\]
For $r\geq 1$, let $G_r$ be the layered directed graph whose vertex layers are
$L_0=\{s\}$, $L_i=\{a_i,b_i\}$ for $1\leq i\leq r$, and
$L_{r+1}=\{t\}$. The edge set $E(G_r)$ contains exactly all directed edges from
one layer to the next, i.e., all pairs $(u,v)$ with $u\in L_i$ and
$v\in L_{i+1}$ for some $0\leq i\leq r$. Let
$\D_r=\{\Edge(u,v)\mid (u,v)\in E(G_r)\}$, with all facts endogenous, and let
$A_r=\Path(s,t)$.

For any $S\subseteq \D_r$, let $G_S$ be the subgraph of $G_r$ whose edge set corresponds to the edge facts in $S$. We prove
\[
\begin{aligned}
  \TC\cup S\models\Path(u,v)
  \,\,\Longleftrightarrow\,\,&
  G_S\text{ has a positive-length directed walk}\\
  &\text{from }u\text{ to }v.
\end{aligned}
\]
For the forward implication, use induction on the first Kleene stage at which $\Path(u,v)$ appears. If the base rule derives $\Path(u,v)$, then $\Edge(u,v)\in S$, giving a one-edge walk. If the recursive rule derives $\Path(u,v)$, then $\Edge(u,w)\in S$ and $\Path(w,v)$ appeared at an earlier stage; by induction, $G_S$ has a positive-length walk from $w$ to $v$, and prefixing the edge $(u,w)$ gives a positive-length walk from $u$ to $v$. For the reverse implication, use induction on the length $\ell\geq 1$ of a walk $u=v_0,v_1,\ldots,v_\ell=v$. If $\ell=1$, the base rule applies to $\Edge(u,v)$. If $\ell>1$, the induction hypothesis derives $\Path(v_1,v)$ from the suffix, and the recursive rule together with $\Edge(u,v_1)$ derives $\Path(u,v)$.

Since $G_r$ is acyclic and layered, every directed walk from $s$ to $t$ is a simple path that uses exactly one edge between each consecutive pair of layers. If $P$ is such a path, then its edge-fact set $F(P)$ entails $\Path(s,t)$. Removing any edge of $P$ leaves no edge sequence from $s$ to $t$ inside the graph with edge set $E(P)\setminus\{e\}$; hence no proper subset of $F(P)$ is a support. Conversely, if $S$ is a minimal support for $A_r$, then $G_S$ contains an $s$-$t$ path $P$. The set $F(P)\subseteq S$ is itself a support, so the minimality of $S$ gives $S=F(P)$. Thus, the minimal supports are exactly the edge-fact sets of simple directed $s$-$t$ paths in $G_r$.

Each such path chooses one of the two vertices in every intermediate layer and has $r+1$ edges, so there are $2^r$ minimal supports and each has cardinality $r+1$. The graph has two edges from $L_0$ to $L_1$, two edges from $L_r$ to $L_{r+1}$, and $4(r-1)$ edges between intermediate consecutive layers. Thus, $|\D_r|=4r$, the maximum support cardinality $\max_{S\in\Supp_{\TC,\D_r,A_r}} |S|$ is $r+1=\Theta(|\D_r|)$, and the number of minimal supports $|\Supp_{\TC,\D_r,A_r}|$ is $2^r=2^{\Omega(|\D_r|)}$.
\end{proof}
Theorem~\ref{thm:separation} separates fixed $\UCQ$s from recursive Datalog by
the size and number of minimal supports.


\section{Causality and Robustness}
\label{sec:causality-supports}

We now connect supports with the standard deletion-based notions of query-answer causality: actual cause, counterfactual cause, contingency, and responsibility \cite{meliou2010causality,halpern2005causes,chockler2004responsibility}. We show that the support hypergraph is exact for these notions and for deletion robustness.

\begin{definition}
\label{def:causality}
Let $\PiP$ be a positive Datalog program, let $(\Dx,\Dn,A)$ be a partitioned
instance for $\PiP$, let $\D=\Dx\cup\Dn$, and assume $\PiP\cup\D\models A$. For
$\tau\in\Dn$, the fact $\tau$ is a counterfactual cause of $A$ if
$\PiP\cup(\D\setminus\{\tau\})\not\models A$. The fact $\tau$ is an actual cause of $A$ if there is $\Gamma\subseteq\Dn\setminus\{\tau\}$ such that $\PiP\cup(\D\setminus\Gamma)\models A$ and $\PiP\cup(\D\setminus(\Gamma\cup\{\tau\}))\not\models A$. Such a set $\Gamma$ is a contingency set for $\tau$. The responsibility of $\tau$ for $A$ is
\[\ma R_{\PiP,\D,A}(\tau)=
\begin{cases}
\dfrac{1}{1+\ma C_{\PiP,\D,A}(\tau)}
& \text{if $\tau$ is an actual cause,}\\[1ex]
0 & \text{otherwise,}
\end{cases}\]
where $\ma C_{\PiP,\D,A}(\tau)$ is the minimum cardinality of a
contingency set for $\tau$.
\end{definition}

 As in resilience and deletion propagation for CQs \cite{freire2015complexity}, we measure deletions by whether they destroy the answer .
\begin{definition}
\label{def:robustness-transversal}
Let $(\Dx,\Dn,A)$ be a partitioned instance for a positive Datalog program  $\PiP$, $\D=\Dx\cup\Dn$ and $\PiP\cup\D\models A$ be as in Definition~\ref{def:causality}. Then,
the deletion robustness radius of $A$ is
\[
  r_{\PiP,\D,A}
  =
  \min\{\,|\Delta| \mid \Delta\subseteq\Dn
  \text{ and }
  \PiP\cup(\D\setminus\Delta)\not\models A\,\}.
\]
If no such $\Delta$ exists, then $r_{\PiP,\D,A}=\Inf$. 
\end{definition}

We next relate the support hypergraph to deletion robustness. 
Let $U$ be a finite set and let $\mathcal{F}\subseteq 2^U$. A transversal of
$\mathcal{F}$ is a set $T\subseteq U$ such that $T\cap F\neq\emptyset$ for
every $F\in\mathcal{F}$ \cite{berge1984hypergraphs}. The transversal number
$\trnum(\mathcal{F})$ is the minimum cardinality of such a set. If
$\emptyset\in\mathcal{F}$, then no transversal exists, and we set
$\trnum(\mathcal{F})=\Inf$. This convention matches exogenous entailment: if
the answer already follows without endogenous facts, then no deletion of
endogenous facts can destroy it.

\begin{theorem}
\label{thm:causality-robustness}
Let $(\Dx,\Dn,A)$ be a partitioned instance for a positive Datalog program  $\PiP$, $\D=\Dx\cup\Dn$ and $\PiP\cup\D\models A$. Let $\Supp=\Supp_{\PiP,\D,A}$. Then, $\forall \tau\in\Dn$:
\begin{enumerate}
    \item $\tau$  is an actual cause of $A$ iff\,  $\exists S\in\Supp.\, \tau \in S$,
    \item  $\tau$  is a counterfactual cause of $A$ iff\,  $\forall S\in\Supp.\, \tau \in S$.
\end{enumerate}
Moreover, $r_{\PiP,\D,A}=\trnum(\Supp)$.
\end{theorem}

\begin{proof}
Suppose $\tau$ is an actual cause with contingency set $\Gamma$. Let $E=\Dn\setminus\Gamma$. The first contingency condition gives $\PiP\cup\Dx\cup E\models A$. By Proposition~\ref{prop:decomposition}, choose $S\in\Supp$ with $S\subseteq E$. If $\tau\notin S$, then $S\subseteq E\setminus\{\tau\}$. Monotonicity of positive Datalog entailment then gives $\PiP\cup\Dx\cup(E\setminus\{\tau\})\models A$, contradicting the second contingency condition. Hence, $\tau\in S$.

Conversely, let $S\in\Supp$ with $\tau\in S$, and put $\Gamma=\Dn\setminus S$. Then $\Gamma\subseteq\Dn\setminus\{\tau\}$, and $\D\setminus\Gamma=\Dx\cup S$, so $\PiP\cup(\D\setminus\Gamma)\models A$. Also,
\[
  \D\setminus(\Gamma\cup\{\tau\})=\Dx\cup(S\setminus\{\tau\}).
\]
Since $S\setminus\{\tau\}\subsetneq S$ and $S$ is inclusion-minimal among supports, $\PiP\cup\Dx\cup(S\setminus\{\tau\})\not\models A$. Thus, $\Gamma$ is a contingency set for $\tau$.

For counterfactual causes, Definition~\ref{def:causality} gives the condition $\PiP\cup\Dx\cup(\Dn\setminus\{\tau\})\not\models A$. By Proposition~\ref{prop:decomposition}, the corresponding entailment holds if and only if some $S\in\Supp$ satisfies $S\subseteq\Dn\setminus\{\tau\}$. Such an $S$ exists if and only if some minimal support avoids $\tau$. Therefore, the non-entailment condition holds if and only if every minimal support contains $\tau$.

For robustness, let $\Delta\subseteq\Dn$ and put $E=\Dn\setminus\Delta$. Proposition~\ref{prop:decomposition} gives
\[
  \PiP\cup(\D\setminus\Delta)\models A
  \,\,\Longleftrightarrow\,\,
  \exists S\in\Supp\text{ with }S\subseteq E.
\]
Since $E=\Dn\setminus\Delta$, the condition $S\subseteq E$ is equivalent to $S\cap\Delta=\emptyset$. Hence, $\Delta$ destroys $A$ if and only if $\Delta$ intersects every member of $\Supp$, i.e., if and only if $\Delta$ is a transversal of $\Supp$. Minimizing $|\Delta|$ over all destroying deletions gives $r_{\PiP,\D,A}=\trnum(\Supp)$, with the case $\emptyset\in\Supp$ covered by the convention $\trnum(\Supp)=\infty$.
\end{proof}

The support hypergraph also determines responsibility. The next proposition gives the exact contingency-size formula. Let $\Supp=\Supp_{\PiP,\D,A}$. For $\tau\in\Dn$, let
\[\Supp^+_{\tau}=\{\,S\in\Supp \mid \tau\in S\,\}
  \quad\text{and}\quad
  \Supp^-_{\tau}=\{\,S\in\Supp\mid \tau\notin S\,\}.\]

\begin{proposition}
\label{prop:responsibility}
For every $\tau\in\Dn$,
\[
\ma R_{\PiP,\D,A}(\tau)=
\begin{cases}
0, & \Supp^+_\tau=\emptyset,\\[.6ex]
\dfrac{1}{1+\ma C_{\PiP,\D,A}(\tau)}, & \Supp^+_\tau\neq\emptyset,
\end{cases}
\]
with $\ma C_{\PiP,\D,A}(\tau)=
\min_{S^+\in\Supp^+_\tau}
\min\{\,|\Gamma| \mid
\Gamma\subseteq\Dn\setminus S^+,\ 
\forall S^-\in\Supp^-_\tau: \Gamma\cap S^-\neq\emptyset\}$.
\end{proposition}

\begin{remark}
Theorem~\ref{thm:causality-robustness} and Proposition~\ref{prop:responsibility}
show that actual causes, counterfactual causes, responsibility, and deletion
robustness are all determined by the support hypergraph.
\end{remark}

\begin{example}
\label{ex:support-hypergraph-quantities}
Suppose that, for a fixed entailment $A$, the minimal supports are
\[ S_1=\{\alpha,\beta\},\qquad
  S_2=\{\beta,\gamma\},\qquad
  S_3=\{\delta,\epsilon\}.\]
Every fact occurring in these supports is an actual cause, and no fact is a
counterfactual cause, because no fact occurs in all three supports. The robustness
value is $2$: one deletion cannot hit all supports, while $\{\beta,\delta\}$
does. The same hypergraph determines responsibility. For instance, $\beta$ has
a one-fact contingency, such as $\{\delta\}$, and hence responsibility $1/2$.
By contrast, $\alpha$ occurs only in $S_1$; to preserve $S_1$ and hit both
supports avoiding $\alpha$, one must delete $\gamma$ and one of $\delta,
\epsilon$. Thus, $\alpha$ has responsibility $1/3$.
\end{example}

\section{Support Hypergraphs and Provenance}
\label{sec:provenance-separation}

We separate support hypergraphs from derivation-based Datalog provenance.
Proof trees, circuits, and semiring expressions record how an answer is derived
\cite{deutch2014circuits,bourgaux2022revisiting,calautti2024complexity}.
Minimal endogenous supports keep only deletion-minimal input sufficiency, which
is exactly the information needed for causality and robustness.

\begin{example}
\label{ex:two-tc-programs}
Consider the left- and right-recursive Datalog formulations of transitive
closure:
\[
\begin{array}{rclcrcl}
\Path(x,y) &\leftarrow& \Edge(x,y)
&\qquad&
\Path(x,y) &\leftarrow& \Edge(x,y),\\
\Path(x,y) &\leftarrow& \Edge(x,z),\Path(z,y)
&&
\Path(x,y) &\leftarrow& \Path(x,z),\Edge(z,y).
\end{array}
\]
The two programs derive the same reachability atoms over every finite
$\EDB$-database, but they decompose paths in different directions and therefore
have different proof-tree shapes. For the fixed goal $\Path(s,t)$, their
minimal endogenous supports are nevertheless identical: they are the edge sets
of simple directed paths from $s$ to $t$. Thus, the support hypergraph follows
the fixed-goal semantics rather than the chosen recursive presentation.
\end{example}

We first formalize the semantic nature of supports. Let $\Pi_1$ and $\Pi_2$
be positive Datalog programs over the same extensional signature, and let $A$
be a ground atom over a common intensional predicate. We say that $\Pi_1$ and
$\Pi_2$ are equivalent for $A$ over finite extensional databases if, for every
finite extensional database $E$,
\[ \Pi_1\cup E\models A
  \quad\Longleftrightarrow\quad
  \Pi_2\cup E\models A,
\]
where each $\Pi_i$ is evaluated over $\adom(\Pi_i,E)$. The syntactic constants of $A$ are fixed on both sides; if one of those constants is outside the active domain of the retained extensional database for a program, then $A$ is not entailed by that program. This is fixed-goal equivalence for $A$, not full program equivalence \cite{abiteboul1995foundations}.

\begin{theorem}\label{thm:program-invariance}
Let $\Pi_1$ and $\Pi_2$ be positive Datalog programs over the same extensional
signature $\EDB$, and let $A$ be a ground atom over a predicate in the common
intensional signature $\IDB$. If $\Pi_1$ and $\Pi_2$ are equivalent for $A$ over finite $\EDB$-databases, then, for every finite $\EDB$-database $\D=\Dx\cup\Dn$ such
that $(\Dx,\Dn,A)$ is a partitioned instance for both programs and
$\Pi_1\cup\D\models A$,
\[\Supp_{\Pi_1,\D,A}=\Supp_{\Pi_2,\D,A}.\]
\end{theorem}

\begin{proof}
For every $S\subseteq\Dn$, the instance $\Dx\cup S$ is a finite extensional database. By equivalence for $A$, $\Pi_1\cup\Dx\cup S\models A$ if and only if $\Pi_2\cup\Dx\cup S\models A$. Hence, the two programs have the same endogenous supports over $(\Dx,\Dn,A)$. Since inclusion-minimality is taken over the same set $\Dn$, the minimal members of the two support families coincide.
\end{proof}

Theorem~\ref{thm:program-invariance} separates minimal endogenous supports from derivation presentations. Positive Datalog programs that are equivalent for the fixed goal may have different rule decompositions and different proof trees, but they induce the same minimal endogenous supports for that goal.

\begin{corollary}
\label{cor:support-invariance}
Let $\Pi_1$ and $\Pi_2$ be positive Datalog programs, let $(\Dx,\Dn,A)$ be a partitioned instance for both programs, let $\D=\Dx\cup\Dn$, and assume
\[
  \Pi_1\cup\D\models A
  \quad\text{and}\quad
  \Pi_2\cup\D\models A.
\]
If  $\Supp_{\Pi_1,\D,A}=\Supp_{\Pi_2,\D,A}$,
then $\Pi_1$ and $\Pi_2$ induce the same actual causes, counterfactual causes, responsibilities, and deletion robustness radius for $A$ over $\D$.
\end{corollary}
Theorem~\ref{thm:causality-robustness} gives the claim for actual causes, counterfactual causes, and deletion robustness. Proposition~\ref{prop:responsibility} gives the claim for responsibility. Both characterizations use only the common support hypergraph.

We next separate supports from proof-tree multiplicity. A finite proof tree for
positive Datalog has the derived atom at the root, rule instances at internal
intensional nodes, and database facts at extensional leaves. We use ordinary finite proof-tree unfoldings; repeated intensional atoms and repeated rule-instance occurrences are not quotiented or represented cyclically.
\begin{proposition}\label{prop:proof-tree-divergence}
There exist a fixed positive recursive Datalog program $\PiP$, a finite
database $\D$ whose facts are all endogenous, and a ground atom $A$ such that
$A$ has infinitely many finite proof trees over $\D$, while
$\Supp_{\PiP,\D,A}$ contains exactly one minimal endogenous support.
\end{proposition}

\begin{proof}
Let $\PiP=\TC$, let $\D=\{\Edge(s,a),\Edge(a,a),\Edge(a,t)\}$,
with all three facts endogenous, and let $A=\Path(s,t)$. For each $k\geq 0$, construct a finite proof tree as follows. The root $\Path(s,t)$ is justified by the ground rule instance
\[ \Path(s,t)\leftarrow \Edge(s,a),\Path(a,t).\]
Below the child $\Path(a,t)$, apply the recursive ground instance
\[\Path(a,t)\leftarrow \Edge(a,a),\Path(a,t)
\]
exactly $k$ times, and then close the final $\Path(a,t)$ node with the base instance $\Path(a,t)\leftarrow \Edge(a,t)$. After the final base-rule closure, every leaf is an extensional fact in $\D$, and every internal node is justified by a ground rule instance of $\TC$. The tree constructed for $k$ has a different number of recursive $\Path(a,t)$ nodes from the tree constructed for $k'\neq k$. Thus, $A$ has infinitely many finite proof trees over $\D$.

We now compute the minimal supports. In the directed graph with edge set corresponding to $\D$, every positive-length directed walk from $s$ to $t$ has the form
\[ s,a,\underbrace{a,\ldots,a}_{k\text{ repetitions}},t
\]
for some $k\geq 0$; equivalently, it uses $\Edge(s,a)$, then $\Edge(a,a)$ zero or more times, and finally $\Edge(a,t)$. Hence, every support for $\Path(s,t)$ contains both $\Edge(s,a)$ and $\Edge(a,t)$. The set $S_0=\{\Edge(s,a),\Edge(a,t)\}$ entails $\Path(s,t)$ by deriving $\Path(a,t)$ from $\Edge(a,t)$ by the base rule, and then $\Path(s,t)$ from $\Edge(s,a)$ and $\Path(a,t)$ by the recursive rule. No proper subset of $S_0$ contains both required edge facts, so no proper subset entails $\Path(s,t)$. Thus, $S_0$ is the unique minimal endogenous support. The loop fact $\Edge(a,a)$ occurs in the proof trees with $k>0$, but it belongs to no minimal support.
\end{proof}

Proposition \ref{prop:proof-tree-divergence} shows that derivation multiplicity and minimal support information can diverge sharply: infinitely many proof trees may collapse to one minimal endogenous support. Thus, support hypergraphs are coarser than proof-tree provenance, but they retain exactly the deletion-relevant information used by causality, responsibility, and robustness.

\section{Reachability}\label{sec:reachability}

We now specialize the support-hypergraph analysis to positive-length
reachability. For the recursive reachability program, minimal endogenous
supports are exactly simple directed paths, and deletion robustness is exactly
the minimum directed edge-cut size.


Let $G=(V,E)$ be a finite directed graph, and let
$\D_G=\{\Edge(u,v)\mid (u,v)\in E\}$ be the corresponding $\EDB$-database. Throughout this section, every fact in $\D_G$ is endogenous. A directed $s$-$t$ path is \emph{simple} if it has the form $P=(v_0,v_1,\ldots,v_k)$, where $v_0=s$, $v_k=t$,
$(v_{i-1},v_i)\in E$ for every $1\leq i\leq k$, and the vertices $v_0,\ldots,v_k$ are pairwise distinct. For such a path $P$, let
\[F(P)=\{\Edge(v_{i-1},v_i)\mid 1\leq i\leq k\}\]
be the corresponding set of edge facts. A directed $s$-$t$ edge cut is a set
$C\subseteq E$ such that $(V,E\setminus C)$ has no directed path from $s$
to $t$. We write $\lambda_G(s,t)$ for the minimum cardinality of such a cut.
\begin{theorem}
\label{thm:reachability}
Let $G=(V,E)$ be a finite directed graph, let $s,t\in V$ be distinct vertices, and assume $\TC\cup\D_G\models\Path(s,t)$, where $\TC$ is the positive-length transitive-closure program. Then,
\[\Supp_{\TC,\D_G,\Path(s,t)} =\{\,F(P)\mid P \text{ is a simple directed }s\text{-}t\text{ path in }G\,\},\]
and \, $r_{\TC,\D_G,\Path(s,t)}=\lambda_G(s,t)$.
\end{theorem}


\begin{proof}
Let $S\subseteq\D_G$, and let $G_S$ be the directed graph with edge set $\{(u,v)\mid \Edge(u,v)\in S\}$. The exact meaning of $\TC$ over $S$ is:
\[
\begin{aligned}
  \TC\cup S\models\Path(a,b)
  \,\,\Longleftrightarrow\,\,&
  G_S\text{ has a positive-length directed walk}\\
  &\text{from }a\text{ to }b.
\end{aligned}
\]
For soundness, argue by induction on the fixpoint stage at which $\Path(a,b)$ first appears. If it is derived by the base rule, then $\Edge(a,b)\in S$, so the one-edge walk exists. If it is derived by the recursive rule, then $\Edge(a,c)\in S$ and $\Path(c,b)$ was derived earlier; by induction there is a positive-length walk from $c$ to $b$, and prefixing $(a,c)$ gives a walk from $a$ to $b$. For completeness, use induction on the length $\ell\geq 1$ of a directed walk $a=v_0,v_1,\ldots,v_\ell=b$. The case $\ell=1$ uses the base rule. For $\ell>1$, the suffix $v_1,\ldots,v_\ell$ derives $\Path(v_1,b)$ by induction, and the recursive rule with $\Edge(a,v_1)$ derives $\Path(a,b)$.

Since $s\neq t$, every positive-length directed walk from $s$ to $t$ contains a simple directed $s$-$t$ path obtained by removing closed repeated-vertex segments. If $P$ is a simple directed $s$-$t$ path, then $F(P)$ entails $\Path(s,t)$ by the equivalence above. No proper subset of $F(P)$ entails $\Path(s,t)$: after removing any edge of the path, the graph containing only the remaining path edges has no directed $s$-$t$ walk. Thus, $F(P)$ is a minimal support.

Conversely, let $S\in\Supp_{\TC,\D_G,\Path(s,t)}$. By the equivalence above, $G_S$ has a positive-length directed walk from $s$ to $t$, and therefore a simple directed $s$-$t$ path $P$. Then, $F(P)\subseteq S$ and $F(P)$ itself is a support. By inclusion-minimality of $S$, we must have $S=F(P)$. This proves the stated characterization of minimal supports.

For robustness, let $\Delta\subseteq\D_G$ and let $C=\{(u,v)\mid\Edge(u,v)\in\Delta\}$. By Theorem~\ref{thm:causality-robustness}, $\Delta$ destroys $\Path(s,t)$ if and only if it intersects every minimal support. By the support characterization, this is equivalent to the condition that $C$ intersects every simple directed $s$-$t$ path. This holds if and only if $(V,E\setminus C)$ has no directed path from $s$ to $t$, i.e., if and only if $C$ is a directed $s$-$t$ edge cut. Minimizing $|\Delta|=|C|$ gives $r_{\TC,\D_G,\Path(s,t)}=\lambda_G(s,t)$. 
\end{proof}

\begin{remark}
Theorem~\ref{thm:causality-robustness} applied to the path-support characterization gives the corresponding causal statements. An edge fact $\Edge(u,v)$ is an actual cause of $\Path(s,t)$ if and only if $(u,v)$ lies on some simple directed $s$-$t$ path, and it is a counterfactual cause if and only if $(u,v)$ lies on every simple directed $s$-$t$ path.
\end{remark}

Theorem~\ref{thm:reachability} complements Theorem~\ref{thm:separation}: although reachability can have exponentially many minimal supports, its deletion robustness is computable as a minimum directed edge cut, without enumerating those supports.

\begin{corollary}\label{cor:compact}
There are partitioned instances $(\emptyset,\D_r,A_r)$ for the fixed program $\TC$ such that
\[|\Supp_{\TC,\D_r,A_r}|=2^{\Omega(|\D_r|)},\]
while the robustness radius $r_{\TC,\D_r,A_r}$ is computable in polynomial time without enumerating the support family.
\end{corollary}

\noindent Corollary \ref{cor:compact} follows from Theorem~\ref{thm:reachability} by assigning unit capacity to every edge and applying a standard maximum-flow/minimum-cut algorithm \cite{ford1956maximal,ahuja2002combinatorial}.
\section{Hardness}
\label{sec:hardness}

We include one compact hardness calibration. The reachability result gives a tractable structured case; the next proposition shows that the threshold problem for deletion robustness is already NP-hard for positive nonrecursive Datalog, because it can express the problem of hitting conjunctive-query witnesses \cite{buneman2002propagation,karp2009reducibility}.

\begin{definition}
\label{def:destroy}
Let $\PiP$ be a positive Datalog program and let $A_0$ be a fixed Boolean intensional atom of $\PiP$. The problem $\Destroy(\PiP,A_0)$ takes as input disjoint finite $\EDB$-databases $\Dx,\Dn$ and an integer $k$, with
$\PiP\cup(\Dx\cup\Dn)\models A_0$, and asks whether there is
$\Delta\subseteq\Dn$ such that $|\Delta|\leq k$ and
\[\PiP\cup\bigl(\Dx\cup(\Dn\setminus\Delta)\bigr)\not\models A_0.
\]
\end{definition}

\begin{proposition}
\label{prop:nphard}
There exists a fixed positive nonrecursive Datalog program $\PiP$ and a fixed Boolean intensional atom $A_0$ such that $\Destroy(\PiP,A_0)$ is NP-hard in data complexity.
\end{proposition}
\begin{proof}
We reduce from \textsc{Vertex Cover} restricted to graphs with at least one edge. This restriction remains NP-hard \cite{karp2009reducibility}: map an arbitrary instance $(G,k)$ to the disjoint union of $G$ with one fresh edge and increase the bound to $k+1$. Given an undirected graph $G=(V,E)$ with $E\neq\emptyset$ and an integer $k$, use a binary $\EDB$ predicate $R$, whose facts will be exogenous, and a unary $\EDB$ predicate $P$, whose facts will be endogenous. For every vertex $u\in V$, include $P(u)$ in $\Dn$. For every undirected edge $\{u,v\}\in E$, choose one orientation $(u,v)$ and include $R(u,v)$ in $\Dx$. Let $\PiP$ consist of the single positive nonrecursive rule \[\Goal \leftarrow R(x,y),P(x),P(y),\]
with $\Goal$ nullary, and set $A_0=\Goal$. Since $E\neq\emptyset$, some fact $R(u,v)\in\Dx$ is present and both $P(u),P(v)\in\Dn$ are initially present. Thus, the promise $\PiP\cup(\Dx\cup\Dn)\models\Goal$ holds.

If $C\subseteq V$ is a vertex cover of size at most $k$, delete $\Delta_C=\{P(u):u\in C\}$. Suppose $\Goal$ were still derivable. Then for some oriented edge fact $R(a,b)\in\Dx$, both $P(a)$ and $P(b)$ would remain. Hence, neither endpoint of the corresponding undirected edge belongs to $C$, contradicting that $C$ is a vertex cover. Thus, $\Delta_C$ destroys $\Goal$, and $|\Delta_C|\leq k$.

Conversely, let $\Delta\subseteq\Dn$ with $|\Delta|\leq k$ destroy $\Goal$, and put $C_\Delta=\{u\in V:P(u)\in\Delta\}$. If some edge $\{u,v\}\in E$ had no endpoint in $C_\Delta$, then for its chosen orientation $(a,b)$, the facts $R(a,b)$, $P(a)$, and $P(b)$ would all remain after deleting $\Delta$. The rule would then derive $\Goal$, contradicting that $\Delta$ destroys it. Therefore, $C_\Delta$ is a vertex cover of size at most $k$. The reduction is polynomial and uses a fixed program and fixed goal atom, so $\Destroy(\PiP,A_0)$ is NP-hard in data complexity.
\end{proof}


\section{Related Work and Discussion}
\label{sec:related}

Database causality studies answers and non-answers through actual
causes, counterfactual causes, contingencies, and responsibility
\cite{meliou2010causality,meliou2010complexity}. Meliou et al. study causes
and responsibility for CQs in data complexity and identify tractable and
intractable cases for responsibility \cite{meliou2010complexity}. Meliou,
Roy, and Suciu relate causality to provenance, deletion propagation, why-not
queries, and database explanations \cite{meliou2014causality}. We use the
same tuple-deletion notions, but we study the support structure induced by
positive Datalog entailments. Our results show that recursion changes this
structure: fixed $\UCQ$s have bounded minimal supports, whereas one fixed
recursive program can have supports of unbounded size and exponential
multiplicity.

Structural-model accounts of causality and responsibility originate in the
work of Halpern and Pearl on actual causality and Chockler and Halpern on
responsibility and blame
\cite{halpern2005causes,chockler2004responsibility}. Database causality
specializes these notions to endogenous tuple interventions
\cite{meliou2010causality,meliou2010complexity}. We keep this intervention
model, but we restrict interventions to deletions of endogenous extensional
facts under a fixed positive Datalog program. Under this monotone deletion
semantics, our support hypergraph exactly determines actual causes,
counterfactual causes, responsibility, and deletion robustness.

Datalog already appears in the causality literature through abduction, view
updates, and integrity constraints. Salimi and Bertossi relate answer causality to Datalog abduction
and view-update problems, including settings with integrity constraints \cite{salimi2015query,salimi2016causes,bertossi2017causes}.
We instead contribute a
support-level analysis: we distinguish nonrecursive and recursive programs by the
size and number of minimal supports, prove invariance under fixed-goal
equivalent programs, and show that all-endogenous positive-length reachability
reduces robustness to a directed edge-cut problem.

Provenance, lineage, and semiring annotations describe how query results depend on input data \cite{buneman2001and,green2007provenance}. Datalog provenance has been studied through
proof trees, circuits, semiring semantics, and the complexity of why-provenance
\cite{deutch2014circuits,bourgaux2022revisiting,calautti2024complexity}.
Our support hypergraphs are coarser than these provenance objects: they discard
rule decompositions, proof-tree multiplicities, and annotations. We make this
loss deliberately. For deletion-based explanation, it suffices to retain exactly
the inclusion-minimal endogenous fact sets that entail the answer.

Deletion propagation and resilience ask how input tuples must be removed to
eliminate query answers or view tuples
\cite{buneman2002propagation,freire2015complexity}. Tiresias and related how-to
systems compute database updates that achieve desired query-output changes
under constraints \cite{meliou2012tiresias}. 
Our robustness question is
narrower: we delete only endogenous facts and ask when one already-derived
positive Datalog entailment disappears. The support hypergraph reduces this
question to a transversal problem.

\section{Conclusion}\label{sec:conclusion}
We have introduced deletion-based explanation for positive Datalog entailments through
inclusion-minimal endogenous supports. We isolated the support hypergraph as a
semantic abstraction that captures the information needed for causality,
responsibility, and deletion robustness, while abstracting from derivation
multiplicity. Our results separate nonrecursive and recursive rule evaluation at
the level of minimal supports. While UCQs yield bounded and polynomially
many minimal supports, positive-length transitive closure can yield supports of
unbounded size and exponential multiplicity. In all-endogenous reachability,
deletion robustness nevertheless coincides with the minimum directed
$s$-$t$ edge cut.

For future work, several directions remain. One is to identify recursive Datalog
fragments that admit compact representations of minimal supports. Another is to
extend the analysis to stratified negation and weighted deletion costs. A third
is to study when provenance circuits can be reduced to support hypergraphs while
preserving deletion-relevant quantities. Answering these questions would clarify the representational and algorithmic scope of support-based explanation in richer rule languages.

\section*{Acknowledgments.}
This work is funded by the German Research Foundation (DFG) -- SFB 1574 Circular Factory-- 471687386.

\section*{Disclosure of Interests}
The authors have no competing interests to declare.
\bibliographystyle{splncs04}
\bibliography{ref}

\vspace{1cm}
\appendix

\section{Extended Proofs}
\label{app:proofs}
We use the notation of the main text: $\PiP$ is a positive Datalog program, $\Dx$ and $\Dn$ are disjoint finite sets of extensional facts, $\D=\Dx\cup\Dn$, and $A$ is a ground atom with $\PiP\cup\D\models A$. Entailment is evaluated over the active-domain grounding determined by the constants occurring in the program and in the retained extensional instance. Whenever we consider a retained set $E\subseteq\Dn$, the associated database is $\Dx\cup E$; hence all support, causality, and robustness arguments are over the same endogenous universe $\Dn$.

\subsection*{Proof of Proposition~\ref{prop:decomposition}}

We first prove that $\Supp_{\PiP,\D,A}$ is a finite antichain. By definition, every member of $\Supp_{\PiP,\D,A}$ is a subset of $\Dn$. Since $\Dn$ is finite, the powerset $2^{\Dn}$ is finite. Hence, $\Supp_{\PiP,\D,A}\subseteq 2^{\Dn}$ is finite. If $S_1,S_2\in\Supp_{\PiP,\D,A}$ and $S_1\subsetneq S_2$, then $S_2$ is not inclusion-minimal among endogenous supports, because $S_1$ is a proper subset of $S_2$ and $\PiP\cup\Dx\cup S_1\models A$. This contradicts the definition of minimal support. Therefore, no two distinct supports in $\Supp_{\PiP,\D,A}$ are comparable by proper inclusion.

We now prove the decomposition equivalence. Let $E\subseteq\Dn$.

Assume first that $\PiP\cup\Dx\cup E\models A$. Consider the family
\[
  \mathcal{F}_E=
  \{\,F\subseteq E\mid \PiP\cup\Dx\cup F\models A\,\}.
\]
This family is nonempty because $E\in\mathcal{F}_E$. It is finite because $E$ is finite. Hence, it contains an inclusion-minimal member; call it $S$. Then $S\subseteq E$ and $\PiP\cup\Dx\cup S\models A$. We must show that $S$ is minimal among all subsets of $\Dn$, not merely among subsets of $E$. Suppose, for a contradiction, that some $S'\subsetneq S$ satisfies $\PiP\cup\Dx\cup S'\models A$. Since $S'\subsetneq S\subseteq E$, we have $S'\in\mathcal{F}_E$, contradicting the choice of $S$ as an inclusion-minimal member of $\mathcal{F}_E$. Thus, $S\in\Supp_{\PiP,\D,A}$ and $S\subseteq E$.

Conversely, assume that there exists $S\in\Supp_{\PiP,\D,A}$ such that $S\subseteq E$. By definition of support, $\PiP\cup\Dx\cup S\models A$. We distinguish the extensional and intensional cases for $A$.

If $A$ is an extensional atom, then the semantics of $\PiP\cup\Dx\cup S\models A$ gives $A\in\Dx\cup S$. Since $S\subseteq E$, every fact in $S$ is also in $E$. Hence, $A\in\Dx\cup E$, and therefore $\PiP\cup\Dx\cup E\models A$.

Suppose now that $A$ is intensional. Let $I^S_0=\emptyset$ and
\[
  I^S_{i+1}=I^S_i\cup T_{\PiP,\Dx\cup S}(I^S_i)
\]
be the Kleene sequence for $\PiP$ over $\Dx\cup S$. Similarly, let $I^E_0=\emptyset$ and
\[
  I^E_{i+1}=I^E_i\cup T_{\PiP,\Dx\cup E}(I^E_i)
\]
be the Kleene sequence over $\Dx\cup E$. Since $S\subseteq E$, the active domain of $\PiP$ with $\Dx\cup S$ is contained in the active domain of $\PiP$ with $\Dx\cup E$. Therefore, every ground rule instance available over $\Dx\cup S$ is also a ground rule instance available over $\Dx\cup E$.

We prove by induction on $i\geq 0$ that $I^S_i\subseteq I^E_i$. For $i=0$, both sets are empty. For the induction step, assume $I^S_i\subseteq I^E_i$, and let $B\in I^S_{i+1}$. By definition of $I^S_{i+1}$, either $B\in I^S_i$ or $B\in T_{\PiP,\Dx\cup S}(I^S_i)$. In the first case, the induction hypothesis gives $B\in I^E_i\subseteq I^E_{i+1}$. In the second case, there is a ground instance of a rule of $\PiP$ with head $B$ such that every extensional body atom is in $\Dx\cup S$ and every intensional body atom is in $I^S_i$. Because $S\subseteq E$ and $I^S_i\subseteq I^E_i$, the same ground rule instance has all extensional body atoms in $\Dx\cup E$ and all intensional body atoms in $I^E_i$. Hence, $B\in T_{\PiP,\Dx\cup E}(I^E_i)\subseteq I^E_{i+1}$. This proves $I^S_{i+1}\subseteq I^E_{i+1}$.

The least fixpoint over a finite active domain is the union of the finite Kleene sequence at its stabilization point. Since $\PiP\cup\Dx\cup S\models A$ and $A$ is intensional, $A\in I^S_j$ for some stage $j$. The inclusion $I^S_j\subseteq I^E_j$ gives $A\in I^E_j$, and therefore $\PiP\cup\Dx\cup E\models A$. The two directions prove the claimed equivalence.

\subsection*{Proof of Theorem~\ref{thm:separation}}

We first prove the claim for fixed UCQs. Let $q$ be a fixed $\UCQ$, and let $m(q)$ be the maximum number of atoms in a disjunct of $q$. Let $\D=\Dx\cup\Dn$ be a finite database, and let $\bar a$ be an answer to $q$. Let $S\subseteq\Dn$ be a minimal endogenous support for $q(\bar a)$.

Since $S$ is a support, $\Dx\cup S\models q(\bar a)$. Hence, some disjunct $q_j$ of $q$ has a satisfying valuation $h$ over $\Dx\cup S$ that maps the free variables of $q_j$ to $\bar a$. For each body atom $R(\bar z)$ of $q_j$, the ground atom $R(h(\bar z))$ belongs to $\Dx\cup S$. Let
\[
  W=\{\,R(h(\bar z))\in S \mid R(\bar z)\text{ is a body atom of }q_j\,\}.
\]
The set $W$ contains exactly the endogenous facts among the facts used by this valuation. The valuation has at most one ground fact for each body atom of $q_j$. Since $q_j$ has at most $m(q)$ atoms, $|W|\leq m(q)$. Every extensional fact required by the valuation is either exogenous, hence in $\Dx$, or endogenous and included in $W$. Therefore, the same valuation witnesses $\Dx\cup W\models q(\bar a)$. Thus, $W$ is a support and $W\subseteq S$. Since $S$ is inclusion-minimal among supports, $S=W$. Hence, $|S|\leq m(q)$.

Every minimal support for $q(\bar a)$ is therefore a subset of $\Dn$ of cardinality at most $m(q)$. The number of such subsets is bounded by
\[
  \sum_{i=0}^{m(q)}\binom{|\Dn|}{i}.
\]
Because $q$ is fixed, $m(q)$ is constant in data complexity, and this expression is polynomial in $|\Dn|$.

We now prove the recursive separation. Let $r\geq 1$. The graph $G_r=(V_r,F_r)$ has layers
\[
  L_0=\{s\},\qquad L_i=\{a_i,b_i\}\text{ for }1\leq i\leq r,
  \qquad L_{r+1}=\{t\},
\]
and contains every directed edge from $L_i$ to $L_{i+1}$ for $0\leq i\leq r$. There are two edges from $L_0$ to $L_1$, two edges from $L_r$ to $L_{r+1}$, and four edges from $L_i$ to $L_{i+1}$ for each $1\leq i<r$. Thus, $|F_r|=2+4(r-1)+2=4r$. Let
\[
  \D_r=\{\,\Edge(u,v)\mid (u,v)\in F_r\,\}
\]
be the all-endogenous database, and let $A_r=\Path(s,t)$.

We need the following reachability characterization for $\TC$. For any $S\subseteq\D_r$, let $H_S$ be the directed graph whose edges are the pairs $(u,v)$ such that $\Edge(u,v)\in S$. We claim that
\[
\begin{aligned}
  \TC\cup S\models\Path(x,y)
  \,\,\,\Longleftrightarrow\,\,\,{}&
  H_S\text{ contains a positive-length directed walk}\\
  &\text{from }x\text{ to }y.
\end{aligned}
\]
For the left-to-right direction, define the Kleene sequence $I_0=\emptyset$ and $I_{i+1}=I_i\cup T_{\TC,S}(I_i)$. We prove by induction on $i$ that every atom $\Path(x,y)\in I_i$ corresponds to a positive-length directed walk from $x$ to $y$ in $H_S$ of length at most $i$. The case $i=0$ is vacuous. For the step, let $\Path(x,y)\in I_{i+1}$. If it already belongs to $I_i$, the induction hypothesis applies. Otherwise, it is produced by a ground rule instance. If the base rule is used, then $\Edge(x,y)\in S$, so the single edge $(x,y)$ is a walk of length $1$. If the recursive rule is used, then for some $z$, $\Edge(x,z)\in S$ and $\Path(z,y)\in I_i$. By the induction hypothesis, there is a walk from $z$ to $y$ of length at most $i$; prefixing the edge $(x,z)$ gives a walk from $x$ to $y$ of length at most $i+1$.

For the right-to-left direction, we prove by induction on the length $\ell\geq 1$ of a directed walk. If $\ell=1$, the walk is one edge $(x,y)$, so $\Edge(x,y)\in S$ and the base rule derives $\Path(x,y)$ at stage $1$. If $\ell>1$, we write the walk as
\[
  x=v_0,v_1,\ldots,v_\ell=y.
\]
The suffix $v_1,\ldots,v_\ell$ has length $\ell-1$. By the induction hypothesis, $\Path(v_1,y)$ is derived by stage $\ell-1$. Since $\Edge(x,v_1)\in S$, the recursive rule derives $\Path(x,y)$ by stage $\ell$.

For distinct endpoints, a positive-length directed walk contains a simple directed path between the same endpoints: whenever a vertex repeats, delete the closed segment between two consecutive occurrences of that vertex; this strictly decreases the walk length and preserves the endpoints. Repeating this operation terminates and yields a walk with no repeated vertices, i.e., a simple directed path.

We now identify the minimal supports for $A_r$. If $P$ is a simple directed $s$-$t$ path in $G_r$, then its edge-fact set $F(P)$ entails $\Path(s,t)$ by the characterization above. No proper subset of $F(P)$ entails $\Path(s,t)$: removing any edge of the path breaks the only directed route available in the graph whose edge set is contained in $F(P)$. Hence, $F(P)$ is a minimal support. Conversely, if $S$ is a minimal support for $A_r$, then $H_S$ contains a positive-length directed walk from $s$ to $t$, and therefore a simple directed $s$-$t$ path $P$. The set $F(P)$ is contained in $S$ and is itself a support. By minimality, $S=F(P)$.

Every simple directed $s$-$t$ path in $G_r$ chooses exactly one vertex from each intermediate layer $L_i$, $1\leq i\leq r$, and then uses one edge between each pair of consecutive layers. Hence, the number of such paths is $2^r$, and every such path has $r+1$ edges. Therefore,
\[
  |\Supp_{\TC,\D_r,A_r}|=2^r,
  \qquad
  \max\{\,|S|\mid S\in\Supp_{\TC,\D_r,A_r}\,\}=r+1.
\]
Since $|\D_r|=|F_r|=4r$, we obtain $r+1=\Theta(|\D_r|)$ and $2^r=2^{\Omega(|\D_r|)}$.

\subsection*{Proof of Theorem~\ref{thm:causality-robustness}}

Let $\mathcal{S}=\Supp_{\PiP,\D,A}$. We prove the three assertions separately.

For actual causes, assume first that $\tau\in\Dn$ is an actual cause. Then there is a contingency set $\Gamma\subseteq\Dn\setminus\{\tau\}$ such that
\[
  \PiP\cup(\D\setminus\Gamma)\models A
  \quad\text{and}\quad
  \PiP\cup(\D\setminus(\Gamma\cup\{\tau\}))\not\models A.
\]
Let $E=\Dn\setminus\Gamma$. Since $\D\setminus\Gamma=\Dx\cup E$, Proposition~\ref{prop:decomposition} gives a support $S\in\mathcal{S}$ with $S\subseteq E$. If $\tau\notin S$, then $S\subseteq E\setminus\{\tau\}$. Applying Proposition~\ref{prop:decomposition} again with $E\setminus\{\tau\}$ gives
\[
  \PiP\cup\Dx\cup(E\setminus\{\tau\})\models A.
\]
But $\Dx\cup(E\setminus\{\tau\})=\D\setminus(\Gamma\cup\{\tau\})$, contradicting the non-entailment condition above. Therefore, $\tau\in S$.

Conversely, suppose that $\tau\in S$ for some $S\in\mathcal{S}$. Let $\Gamma=\Dn\setminus S$. Since $\tau\in S$, we have $\Gamma\subseteq\Dn\setminus\{\tau\}$. Also, $\D\setminus\Gamma=\Dx\cup S$, so $\PiP\cup(\D\setminus\Gamma)\models A$. After also deleting $\tau$, the retained endogenous set is $S\setminus\{\tau\}$. This is a proper subset of $S$, and $S$ is a minimal support. Hence,
\[
  \PiP\cup\Dx\cup(S\setminus\{\tau\})\not\models A.
\]
Since $\Dx\cup(S\setminus\{\tau\})=\D\setminus(\Gamma\cup\{\tau\})$, the set $\Gamma$ is a contingency for $\tau$. Therefore, $\tau$ is an actual cause.

For counterfactual causes, by definition $\tau$ is counterfactual exactly when
\[
  \PiP\cup\Dx\cup(\Dn\setminus\{\tau\})\not\models A.
\]
By Proposition~\ref{prop:decomposition}, the corresponding entailment holds exactly when there is a support $S\in\mathcal{S}$ with $S\subseteq\Dn\setminus\{\tau\}$. This containment is equivalent to $\tau\notin S$. Hence, non-entailment holds exactly when no minimal support avoids $\tau$, i.e., exactly when every $S\in\mathcal{S}$ contains $\tau$.

For robustness, let $\Delta\subseteq\Dn$ and put $E=\Dn\setminus\Delta$. By Proposition~\ref{prop:decomposition},
\[
  \PiP\cup(\D\setminus\Delta)
  =\PiP\cup\Dx\cup E
  \models A
\]
holds exactly when some $S\in\mathcal{S}$ satisfies $S\subseteq E$. Since $E=\Dn\setminus\Delta$, the inclusion $S\subseteq E$ is equivalent to $S\cap\Delta=\emptyset$. Therefore, $\Delta$ destroys $A$ exactly when $\Delta\cap S\neq\emptyset$ for every $S\in\mathcal{S}$. Such sets $\Delta$ are precisely the transversals of $\mathcal{S}$. Minimizing their cardinality gives
\[
  r_{\PiP,\D,A}=\trnum(\mathcal{S}).
\]
If $\emptyset\in\mathcal{S}$, then no $\Delta\subseteq\Dn$ intersects every support, because no set intersects $\emptyset$. In that case both the deletion robustness radius and the transversal number are $\Inf$ (by the conventions as in the main text).

\subsection*{Proof of Proposition~\ref{prop:responsibility}}

Let
\[ \mathcal{S}^+_\tau=\{\,S\in\Supp_{\PiP,\D,A}\mid\tau\in S\,\},
  \qquad
  \mathcal{S}^-_\tau=\{\,S\in\Supp_{\PiP,\D,A}\mid\tau\notin S\,\}.
\]
If $\mathcal{S}^+_\tau=\emptyset$, then Theorem~\ref{thm:causality-robustness} implies that $\tau$ is not an actual cause. By Definition~\ref{def:causality}, its responsibility is then $0$.

Assume $\mathcal{S}^+_\tau\neq\emptyset$. We characterize exactly the contingency sets for $\tau$. Let $\Gamma\subseteq\Dn\setminus\{\tau\}$. The set $\Gamma$ is a contingency for $\tau$ iff both conditions below hold:
\[
  \PiP\cup(\D\setminus\Gamma)\models A,
  \qquad
  \PiP\cup(\D\setminus(\Gamma\cup\{\tau\}))\not\models A.
\]
By Proposition~\ref{prop:decomposition}, the first condition holds iff there exists a support $S\in\Supp_{\PiP,\D,A}$ such that $S\cap\Gamma=\emptyset$. The second condition holds iff no support is contained in $\Dn\setminus(\Gamma\cup\{\tau\})$. A support containing $\tau$ is never contained in $\Dn\setminus\{\tau\}$, so only supports not containing $\tau$ can violate the second condition. Thus, the second condition is equivalent to
\[
  \Gamma\cap S^-\neq\emptyset
  \quad\text{for every }S^-\in\mathcal{S}^-_\tau.
\]
If this hitting condition holds, then no support in $\mathcal{S}^-_\tau$ can be disjoint from $\Gamma$. Therefore, the support preserved by the first condition must be some $S^+\in\mathcal{S}^+_\tau$. Preserving this support means $S^+\cap\Gamma=\emptyset$, equivalently $\Gamma\subseteq\Dn\setminus S^+$.

Consequently, the contingency sets for $\tau$ are exactly the sets $\Gamma$ for which there exists $S^+\in\mathcal{S}^+_\tau$ such that
\[
  \Gamma\subseteq\Dn\setminus S^+,
  \qquad
  \Gamma\cap S^-\neq\emptyset\text{ for every }S^-\in\mathcal{S}^-_\tau.
\]
The smallest contingency size is therefore the minimum stated in the proposition. Substituting this value into the database-causality definition of responsibility gives the stated formula.

\subsection*{Proof of Theorem~\ref{thm:program-invariance} and Corollary~\ref{cor:support-invariance}}
We first prove Theorem~\ref{thm:program-invariance}. Let $S\subseteq\Dn$ be
arbitrary. Since $\Dx$ and $\Dn$ are finite $\EDB$-databases, the retained
instance $\Dx\cup S$ is a finite $\EDB$-database. By the fixed-goal equivalence
assumption, each program is evaluated over its own active-domain grounding, and
the following equivalence holds for this retained database:
\[
  \PiP_1\cup\Dx\cup S\models A
  \,\,\,\Longleftrightarrow\,\,\,
  \PiP_2\cup\Dx\cup S\models A .
\]
By Definition~\ref{def:support}, the left-hand side says exactly that $S$ is an endogenous support for $A$ with respect to $\PiP_1$ over the partition $(\Dx,\Dn)$. The right-hand side says exactly that the same set $S$ is an endogenous support for $A$ with respect to $\PiP_2$ over the same partition. Hence, for every $S\subseteq\Dn$, $S$ belongs to the full support family of $\PiP_1$ iff it belongs to the full support family of $\PiP_2$.

Let
\[
  \mathcal{T}_i=\{\,S\subseteq\Dn \mid \PiP_i\cup\Dx\cup S\models A\,\}
  \qquad(i\in\{1,2\})
\]
be the full family of endogenous supports before taking inclusion-minimal members. The previous paragraph proves $\mathcal{T}_1=\mathcal{T}_2$. Minimal endogenous supports are precisely the inclusion-minimal members of these finite families, and inclusion is taken over the same universe $\Dn$ in both cases. Therefore, the sets of inclusion-minimal members are equal:
\[
  \Supp_{\PiP_1,\D,A}=\Supp_{\PiP_2,\D,A}.
\]
This proves Theorem~\ref{thm:program-invariance}.

We now prove Corollary~\ref{cor:support-invariance}. Assume
$\Supp_{\PiP_1,\D,A}=\Supp_{\PiP_2,\D,A}=\mathcal{S}$.
Both programs are evaluated over the same endogenous universe $\Dn$. For any $\tau\in\Dn$, Theorem~\ref{thm:causality-robustness} characterizes actual causality by the condition
\[
  \exists S\in\mathcal{S}\text{ such that }\tau\in S,
\]
and counterfactual causality by the condition
$\forall S\in\mathcal{S},\ \tau\in S$. Since these two conditions mention only $\tau$ and the common family $\mathcal{S}$, the two programs have the same actual causes and the same counterfactual causes.

The same theorem gives the deletion robustness radius as the transversal number of the common support family:
\[
  r_{\PiP_i,\D,A}=\trnum(\mathcal{S})\qquad(i\in\{1,2\}).
\]
Thus, the robustness radii coincide. It remains only to check responsibility. For any fixed $\tau\in\Dn$, Proposition~\ref{prop:responsibility} uses:
\[
  \mathcal{S}^{+}_{\tau}=\{S\in\mathcal{S}\mid\tau\in S\},
  \qquad
  \mathcal{S}^{-}_{\tau}=\{S\in\mathcal{S}\mid\tau\notin S\},
\]
and the same ambient set $\Dn$ to compute the minimum contingency size. Since $\mathcal{S}$ and $\Dn$ are common to both programs, the feasible contingency sets and their minimum cardinality are the same for $\PiP_1$ and $\PiP_2$. Therefore, the responsibility value of every $\tau\in\Dn$ is the same for both programs. This proves Corollary~\ref{cor:support-invariance}.

\subsection*{Proof of Proposition~\ref{prop:proof-tree-divergence}}

Let $\PiP=\TC$ and
\[
  \D=\{\Edge(s,a),\Edge(a,a),\Edge(a,t)\},
\]
with all facts endogenous. Let $A=\Path(s,t)$. We use ordinary finite proof trees: each internal node is justified by one ground rule instance, leaves are extensional facts, and repeated occurrences of
the same intensional atom are not identified.

For every integer $k\geq 0$, we construct a finite proof tree for $\Path(s,t)$. The root is labeled $\Path(s,t)$ and is justified by the ground recursive rule instance
\[
  \Path(s,t)\leftarrow \Edge(s,a),\Path(a,t).
\]
Thus, the root has one extensional child labeled $\Edge(s,a)$ and one intensional child labeled $\Path(a,t)$. If $k=0$, close this intensional child by the base rule instance
\[
  \Path(a,t)\leftarrow \Edge(a,t).
\]
If $k>0$, expand the node labeled $\Path(a,t)$ by the recursive rule instance
\[
  \Path(a,t)\leftarrow \Edge(a,a),\Path(a,t)
\]
exactly $k$ consecutive times. After these $k$ recursive expansions, close the remaining node labeled $\Path(a,t)$ by the base rule instance $\Path(a,t)\leftarrow\Edge(a,t)$. After the final base-rule closure, every leaf of the finite tree is one of the extensional facts in $\D$: one occurrence of $\Edge(s,a)$, $k$ occurrences of $\Edge(a,a)$, and one occurrence of $\Edge(a,t)$. Hence, each tree is a valid finite proof tree. The tree for $k+1$ has one more recursive expansion below $\Path(a,t)$ than the tree for $k$, so the trees have distinct depths. Therefore, there are infinitely many finite proof trees.

We now compute the minimal endogenous supports. The set
\[
  S_0=\{\Edge(s,a),\Edge(a,t)\}
\]
entails $\Path(s,t)$: first the base rule derives $\Path(a,t)$ from $\Edge(a,t)$, and then the recursive rule derives $\Path(s,t)$ from $\Edge(s,a)$ and $\Path(a,t)$. No proper subset of $S_0$ entails $\Path(s,t)$. If $\Edge(s,a)$ is absent, no positive-length walk can leave $s$. If $\Edge(a,t)$ is absent, no positive-length walk can reach $t$. The remaining fact $\Edge(a,a)$ is only a self-loop at $a$ and cannot connect $s$ to $t$ without both $\Edge(s,a)$ and $\Edge(a,t)$.

Conversely, any support for $\Path(s,t)$ must contain both facts $\Edge(s,a)$ and $\Edge(a,t)$. Indeed, every positive-length directed walk from $s$ to $t$ in the directed graph with edge set corresponding to $\D$ has the form
\[
  s,a,\underbrace{a,\ldots,a}_{\text{zero or more repetitions}},t,
\]
and therefore uses $\Edge(s,a)$ and $\Edge(a,t)$. Hence, the unique minimal endogenous support is
\[
  \Supp_{\TC,\D,\Path(s,t)}=\{\,S_0\,\}.
\]
The loop fact $\Edge(a,a)$ occurs in the proof trees constructed above when $k>0$, but it is not contained in the unique minimal support.

\subsection*{Proof of Theorem~\ref{thm:reachability}}

Let $G=(V,E)$ be a finite directed graph, let $\D_G=\{\Edge(u,v):(u,v)\in E\}$, and let $S\subseteq\D_G$. Let $G_S=(V,E_S)$ be the directed graph whose edge set corresponds to $S$, where
\[
  E_S=\{\,(u,v)\mid\Edge(u,v)\in S\,\}.
\]
Let $I_0=\emptyset$ and $I_{i+1}=I_i\cup T_{\TC,S}(I_i)$ be the Kleene sequence for $\TC$ over the extensional database $S$.

We first prove soundness of the fixpoint computation with respect to directed walks. For every $i\geq 0$, if $\Path(a,b)\in I_i$, then $G_S$ contains a positive-length directed walk from $a$ to $b$ of length at most $i$. The proof is by induction on $i$. For $i=0$, $I_0=\emptyset$, so there is nothing to prove. Assume the claim holds at stage $i$, and let $\Path(a,b)\in I_{i+1}$. If $\Path(a,b)\in I_i$, the induction hypothesis gives the required walk. Otherwise, $\Path(a,b)$ is newly produced by $T_{\TC,S}(I_i)$. If the base rule produced it, then $\Edge(a,b)\in S$, so $(a,b)\in E_S$ is a walk of length $1\leq i+1$. If the recursive rule produced it, then for some constant $c$, $\Edge(a,c)\in S$ and $\Path(c,b)\in I_i$. The induction hypothesis gives a walk from $c$ to $b$ of length at most $i$. Prefixing the edge $(a,c)$ gives a walk from $a$ to $b$ of length at most $i+1$.

We next prove completeness. For every positive-length directed walk from $a$ to $b$ in $G_S$ of length $\ell\geq 1$, the atom $\Path(a,b)$ belongs to $I_\ell$. The proof is by induction on $\ell$. If $\ell=1$, the walk is the edge $(a,b)$, so $\Edge(a,b)\in S$, and the base rule gives $\Path(a,b)\in I_1$. If $\ell>1$, let the walk as
\[
  a=v_0,v_1,\ldots,v_\ell=b.
\]
The suffix $v_1,\ldots,v_\ell$ is a positive-length directed walk from $v_1$ to $b$ of length $\ell-1$. By the induction hypothesis, $\Path(v_1,b)\in I_{\ell-1}$. Since $\Edge(a,v_1)\in S$, the recursive rule derives $\Path(a,b)$ at stage $\ell$, so $\Path(a,b)\in I_\ell$.

The two inductions prove
\[
\begin{aligned}
  \TC\cup S\models\Path(a,b)
  \,\,\,\Longleftrightarrow\,\,\,{}&
  G_S\text{ contains a positive-length directed walk}\\
  &\text{from }a\text{ to }b.
\end{aligned}
\]
When $a\neq b$, every positive-length directed walk from $a$ to $b$ contains a simple directed $a$-$b$ path. To see this, if a walk repeats a vertex, remove the segment between two equal occurrences of that vertex. The remaining sequence is still a directed walk with the same endpoints and strictly smaller length. Repeating this finite shortening process yields a walk with no repeated vertices, i.e., a simple directed path.

We now specialize to distinct vertices $s,t\in V$. Let $P$ be a simple
directed path from $s$ to $t$ in $G$, and let
$F(P)=\{\Edge(u,v)\mid (u,v)\text{ is an edge of }P\}$. By completeness,
$F(P)$ entails $\Path(s,t)$. We prove that $F(P)$ is inclusion-minimal. Assume the path as
\[
  s=v_0,v_1,\ldots,v_m=t,
\]
where the vertices are pairwise distinct and the path edges are $(v_{j-1},v_j)$ for $1\leq j\leq m$. If one edge $(v_{j-1},v_j)$ is removed from this path-edge set, then no edge remains that goes from any vertex among $v_0,\ldots,v_{j-1}$ to any vertex among $v_j,\ldots,v_m$ inside the subgraph whose edges are the remaining path edges. Therefore, there is no directed walk from $s$ to $t$ using only the remaining path edges. By the equivalence above, the remaining facts do not entail $\Path(s,t)$. Hence, every proper subset obtained by deleting at least one path edge is not a support, and $F(P)$ is a minimal support.

Conversely, let $S\in\Supp_{\TC,\D_G,\Path(s,t)}$. Since $S$ is a support, $\TC\cup S\models\Path(s,t)$. By the equivalence above, $G_S$ contains a positive-length directed walk from $s$ to $t$. Since $s\neq t$, this walk contains a simple directed $s$-$t$ path $P$. All edge facts of $P$ belong to $S$, so $F(P)\subseteq S$. By completeness, $F(P)$ is a support. Since $S$ is inclusion-minimal among supports, $S=F(P)$. This proves
\[
  \Supp_{\TC,\D_G,\Path(s,t)}
  =\{\,F(P)\mid P\text{ is a simple directed }s\text{-}t\text{ path in }G\,\}.
\]

It remains to prove the cut identity. Let $\Delta\subseteq\D_G$, and let
\[
  C=\{\,(u,v)\mid\Edge(u,v)\in\Delta\,\}\subseteq E.
\]
By Theorem~\ref{thm:causality-robustness}, $\Delta$ destroys $\Path(s,t)$ iff it intersects every minimal support. By the support characterization just proved, this holds iff $C$ intersects the edge set of every simple directed $s$-$t$ path. This is equivalent to saying that $E\setminus C$ contains no directed $s$-$t$ path: if a simple path avoids $C$, then it remains after deleting $C$; conversely, if a directed $s$-$t$ path remains after deleting $C$, it contains a simple directed $s$-$t$ subpath that also avoids $C$. Therefore, $C$ is a directed $s$-$t$ edge cut exactly when $\Delta$ destroys $\Path(s,t)$. Since the correspondence $\Delta\mapsto C$ preserves cardinality, minimizing $|\Delta|$ over destroying deletion sets gives
\[ r_{\TC,\D_G,\Path(s,t)}=\lambda_G(s,t).
\]
This proves the support characterization and the cut identity in
Theorem~\ref{thm:reachability}.

\subsection*{Proof of Corollary~\ref{cor:compact}}

We use the layered instances $(\emptyset,\D_r,A_r)$ constructed in the proof of Theorem~\ref{thm:separation}. In that construction, $A_r=\Path(s,t)$, all edge facts are endogenous, and $\D_r$ is the set of edge facts of the layered graph $G_r$. The graph has two edges from the source layer to the first intermediate layer, two edges from the last intermediate layer to the target, and four edges between each consecutive pair of intermediate layers. Hence, the number of endogenous facts is
\[
  |\D_r|=4r.
\]
Thus, polynomial time in the size of $G_r$ is polynomial time in the database
size, because $\D_r$ contains exactly one fact for each edge of $G_r$.

Each simple directed $s$-$t$ path in $G_r$ chooses exactly one of the two vertices in each of the $r$ intermediate layers. These choices are independent, so $G_r$ has $2^r$ simple directed $s$-$t$ paths. By Theorem~\ref{thm:reachability}, the minimal endogenous supports of $A_r$ are exactly the edge-fact sets of these simple paths. Therefore,
\[
  |\Supp_{\TC,\D_r,A_r}|=2^r=2^{\Omega(|\D_r|)}.
\]

The same theorem gives $r_{\TC,\D_r,A_r}=\lambda_{G_r}(s,t)$, where
$\lambda_{G_r}(s,t)$ is the minimum cardinality of a directed $s$-$t$ edge cut.
Assign unit capacity to every edge of $G_r$. By the maximum-flow/minimum-cut
theorem, the value of a maximum $s$-$t$ flow equals the minimum cardinality of a
directed $s$-$t$ edge cut \cite{ford1956maximal,ahuja2002combinatorial}.
Standard maximum-flow algorithms compute this value in time polynomial in the
number of vertices and edges of $G_r$, and hence in time polynomial in
$|\D_r|$. This computation uses $G_r$ as a graph instance and never enumerates
its $2^r$ simple paths. Consequently, the robustness value is computable in
polynomial time even though the support family is exponential.

\subsection*{Proof of Proposition~\ref{prop:nphard}}
We give a many-one reduction from \textsc{Vertex Cover}. The source problem
asks, given an undirected graph $G=(V,E)$ and an integer $k$, whether there is a
set $C\subseteq V$ with $|C|\leq k$ such that every edge in $E$ has at least one
endpoint in $C$. We use the restriction to graphs with at least one edge. This
restriction remains NP-hard. Indeed, map an arbitrary instance $(G,k)$ to the
disjoint union of $G$ with one fresh edge $\{p,q\}$ and replace the bound by
$k+1$; Vertex Cover is NP-complete \cite{karp2009reducibility}. If $G$ has a
vertex cover of size at most $k$, then adding one of $p,q$ gives a cover of the
new graph of size at most $k+1$.

Conversely, any cover of the new graph of size at most $k+1$ must contain at
least one of $p,q$ to cover the fresh edge. Removing one selected endpoint of
that fresh edge from the cover leaves a vertex cover of the original graph $G$
of size at most $k$.

Let $(G,k)$ be an instance of this restricted problem, with $G=(V,E)$ and
$E\neq\emptyset$. We construct an instance of $\Destroy(\PiP,A_0)$. The fixed
positive nonrecursive Datalog program $\PiP$ uses a binary extensional predicate
$R$, a unary extensional predicate $P$, and a nullary intensional predicate
$\Goal$. It consists of the single rule
\[
  \Goal\leftarrow R(x,y),P(x),P(y).
\]
The fixed goal atom is $A_0=\Goal$.

The extensional database is built as follows. For every vertex $u\in V$, put the fact $P(u)$ into the endogenous database $\Dn$. 
For every undirected edge $\{u,v\}\in E$, choose one orientation, denoted
$(a,b)$, and put the fact $R(a,b)$ into the exogenous database $\Dx$. No other
facts are added.

The construction is polynomial in $|V|+|E|$, and the program $\PiP$ and goal $A_0$ do not depend on the input graph.

The constructed instance satisfies the promise in Definition~\ref{def:destroy}. 
Since $E\neq\emptyset$, choose an edge $\{u,v\}\in E$ and let $(a,b)$ be its
chosen orientation.
 Then $R(a,b)\in\Dx$, and both $P(a)$ and $P(b)$ belong to $\Dn$. Therefore, the ground rule instance
\[
  \Goal\leftarrow R(a,b),P(a),P(b)
\]
has all body atoms true in $\Dx\cup\Dn$, so $\PiP\cup(\Dx\cup\Dn)\models\Goal$.

We prove correctness of the reduction in both directions. Suppose first that $G$ has a vertex cover $C\subseteq V$ with $|C|\leq k$. Let
\[
  \Delta_C=\{\,P(u)\mid u\in C\,\}\subseteq\Dn.
\]
Then, $|\Delta_C|=|C|\leq k$. We show that deleting $\Delta_C$ destroys $\Goal$. Assume, toward a contradiction, that
\[
  \PiP\cup(\Dx\cup(\Dn\setminus\Delta_C))\models\Goal.
\]
The only rule with head $\Goal$ is the rule given above. Hence, some ground instance of that rule must have all body atoms true in the retained database. Thus, there are constants $a,b$ such that $R(a,b)\in\Dx$, $P(a)\in\Dn\setminus\Delta_C$, and $P(b)\in\Dn\setminus\Delta_C$.
The fact $R(a,b)$ was introduced from the undirected edge $\{a,b\}\in E$.
 Since $P(a)$ and $P(b)$ were not deleted, neither $a$ nor $b$ belongs to $C$. This contradicts that $C$ covers every edge of $G$. Therefore,
\[
  \PiP\cup(\Dx\cup(\Dn\setminus\Delta_C))\not\models\Goal,
\]
and the constructed $\Destroy(\PiP,A_0)$ instance is a yes-instance.

Conversely, suppose there is a deletion set $\Delta\subseteq\Dn$ with $|\Delta|\leq k$ such that
\[
  \PiP\cup(\Dx\cup(\Dn\setminus\Delta))\not\models\Goal.
\]
Let $C_\Delta=\{\,u\in V\mid P(u)\in\Delta\,\}$.
Since the facts $P(u)$ are in one-to-one correspondence with vertices $u\in V$, we have $|C_\Delta|=|\Delta|\leq k$. We show that $C_\Delta$ is a vertex cover. Assume not. 
Then some edge $\{u,v\}\in E$ has neither endpoint in $C_\Delta$.
 Let $(a,b)$ be the orientation of this edge used in the construction. Since neither endpoint belongs to $C_\Delta$, both $P(a)$ and $P(b)$ remain in $\Dn\setminus\Delta$. Also, $R(a,b)\in\Dx$ by construction. Hence, the ground instance
\[
  \Goal\leftarrow R(a,b),P(a),P(b)
\]
has all body atoms true in $\Dx\cup(\Dn\setminus\Delta)$. It follows that $\PiP\cup(\Dx\cup(\Dn\setminus\Delta))\models\Goal$, contradicting the choice of $\Delta$. Therefore, every edge of $G$ has an endpoint in $C_\Delta$, and $C_\Delta$ is a vertex cover of size at most $k$.

The two directions establish that $(G,k)$ is a yes-instance of \textsc{Vertex Cover} iff the constructed instance is a yes-instance of $\Destroy(\PiP,A_0)$. The reduction is polynomial, while $\PiP$ and $A_0$ are fixed. Hence, $\Destroy(\PiP,A_0)$ is NP-hard in data complexity.
\end{document}